\documentclass[11pt]{article}
\pdfoutput=1
\usepackage[margin=1in]{geometry}
\usepackage{amssymb}
\usepackage{graphicx}
\usepackage{textcomp}
\usepackage{enumerate}
\usepackage{amsfonts,amsmath}
\usepackage{url}
\usepackage[pagewise]{lineno}
\usepackage{authblk}
\usepackage{tikz}
\usepackage{amsthm}
\usepackage{yhmath}
\usepackage{wrapfig}
\usepackage[pdfpagelabels,plainpages=false,hypertexnames=false]{hyperref}%
\hypersetup{%
   breaklinks,%
   colorlinks=true,%
   linkcolor=[rgb]{0.45,0.0,0.0},%
   urlcolor=[rgb]{0.05,0.390,0.0},%
   citecolor=[rgb]{0,0,0.45}
}%

\newtheorem{lemma}{Lemma}[section]

\newtheorem{obs}{Observation}[section]
\newtheorem{cor}{Corollary}
\newtheorem{theorem}{Theorem}

\title{Voronoi Game on Graphs\thanks{A preliminary version of this article appeared in WALCOM-2013 \cite{BandyapadhyayBDS13}.}}

\author[1]{Sayan Bandyapadhyay\thanks{The work was done while the author was a student in Indian Statistical Institute.}\thanks{sayan-bandyapadhyay@uiowa.edu}}
\author[2]{Aritra Banik\thanks{aritrabanik@gmail.com}}
\author[2]{Sandip Das\thanks{sandip.das.69@gmail.com}}
\author[2]{Hirak Sarkar\thanks{hiraksarkar.cs@gmail.com}}

\affil[1]{University of Iowa\\ Iowa City, USA}
\affil[2]{Indian Statistical Institute\\Kolkata, India}

\begin{document}
\maketitle

\begin{abstract}
\textit{Voronoi game} is a geometric model of competitive facility location problem played between two players. Users are generally modeled as points uniformly distributed on a given underlying space. Each player chooses a set of points in the underlying space to place their facilities. Each user avails service from its nearest facility. Service zone of a facility consists of the set of users which are closer to it than any other facility. Payoff of each player is defined by the quantity of users served by all of its facilities. The objective of each player is to maximize their respective payoff. In this paper we consider the two players {\it Voronoi game} where the underlying space is a road network modeled by a graph. In this framework we consider the problem of finding $k$ optimal facility locations of Player 2 given any placement of $m$ facilities by Player 1. Our main result is a dynamic programming based polynomial time algorithm for this problem on tree network. On the other hand, we show that the problem is strongly $\mathcal{NP}$-complete for graphs. This proves that finding a winning strategy of P2 is $\mathcal{NP}$-complete. Consequently, we design an $1-\frac{1}{e}$ factor approximation algorithm, where $e \approx 2.718$.

\end{abstract}




\section{Introduction}\label{sec:intro}

In Competitive facility location problem several market players compete with each other for placing facilities (post office, shopping mall etc.) (\cite{CFL1},\cite{CFL4},\cite{CFL6_Hakimi},\cite{CFL3_Hakimi},\cite{CFL5_Hakimi},\cite{CFL2}). The customers choose the best facility to get services with respect to some specific requirements. The goal is to attract as much customers as possible. For a comprehensive study see the surveys (\cite{CFL_survey2},\cite{CFL_survey1}).

Competitive facility location can also be viewed from the perspective of Game theory. Here in each move a market player places her facilities judiciously so that her {\it Gain} or {\it Payoff} is maximized. An interesting direction is to study how the decision of these players affect each other. Thus game theoretic arguments are used to analyze the best move or winning strategy of the players. 

Ahn {\em et al.}~\cite{Ahn_et_al} consider a competitive facility location problem which they call the \textit{Voronoi Game}. There are only two players P1 and P2 who play this game against each other. Both of the players place a specified number, $m$, of facilities alternately, starting with P1 ($m$ round game). The facilities are placed in a planar region $U$. After placement of all the 2$m$ facilities the nearest neighbor Voronoi diagram of those 2$m$ points is computed and the Voronoi region corresponding to each facility is assigned to it as its service zone. Service zone of a player is the union of the service zones corresponding to its $m$ facilities. The player whose service zone is having larger area wins the game.

Considering the complications of the planar version Ahn {\em et al.}~\cite{Ahn_et_al} focus on an one-dimensional version of this game, where the region is a line segment or a circular arc. They show that the second player always has a winning strategy for this version. They have also considered another version of the game, where instead of placing the facilities alternately P1 places its $m$ facilities at first and then P2 places its $m$ facilities (one-round game). They show that in this case the first player always has a winning strategy. The one-round planar version has studied by Cheong {\em et al.}~\cite{Cheong_One-Round_Voronoi_Game} for a square-shaped region. In this case also the second player always has a winning strategy. Fekete {\em et al.}~\cite{Fekete_one-round_Voronoi_game} have studied the planar one-round version for a rectangular region with aspect ratio $\rho$. They have shown that the second player has a winning strategy for $m\geq 3$ and $\rho > \frac{\sqrt{2}}{m}$, and for $m=2$ and $\rho > \frac{\sqrt{3}}{2}$. The first player wins in all the remaining cases. 

In real life scenario often the facilities like shopping malls are allowed to be placed only on (or beside) road networks (\cite{CFL6_Hakimi},\cite{Hakimi_et_al},\cite{CFL_duopoly},\cite{CFL_tree1},\cite{CFL_survey1}). The customers are also assumed to be on (or beside) the road network for the sake of reachability. The customers always choose their nearest (along the edges of the road network) facility. The problem of interest is to find the placement location of the facilities that attract maximum number of customers. Teramoto, Demaine and Uehara \cite{Voronoi_game_on_graphs} and Durr {\em et al.}~\cite{nash_DurrT07} independently consider this model which they call \textit{discrete} Voronoi game. Here the road network is modeled using a weighted graph. Two players alternately occupy 2$n$ vertices of the graph. Each vertex is assigned to the player who occupies the nearest (with respect to shortest path distance) vertex to it. Either a player dominates larger number of vertices or the game ends in a tie. They have studied the game on complete $k$-ary tree. They show that P1 has a winning strategy if (1) 2$n \leq k$, or (2) $k$ is odd and the complete $k$-ary tree contains at least $(k^3n^2-1)/(k-1)$ vertices. In contrast, in case when $k$ is even, $2n > k$, and the complete $k$-ary tree contains at least $(k^3n^2-1)/(k-1)$ vertices, two players tie if they play optimally. They also consider a restricted version of the game where P1 occupies only one vertex and P2 occupies $n$ vertices. Surprisingly for this case they have shown that it is $\mathcal{NP}$-complete to determine whether P2 has a winning strategy. Moreover, they show that for a given graph $G$ and the number $n$ of turns it is $PSPACE$-complete to determine whether P1 has a winning strategy. Kiyomi, Saitoh and Uehara \cite{Vgameonpath} consider discrete Voronoi game on paths. They show that if the length of the path is even and the number of rounds is even then P1 has a trivial winning strategy. In all the other cases the game ends in a tie. Existence of pure Nash equilibrium has also been studied on this model (\cite{nash_DurrT07},\cite{nash_transitive_FeldmannMM09},\cite{nash_cycle_MavronicolasMPS08}). 

In this paper we study a natural extension of discrete Voronoi game. The game is played on a graph embedded in $\mathbb{R}^2$ whose vertices and edges are having non-negative weights. In this model the facilities can be placed either on the vertices or on the points of the edges. At first P1 places $m$ facilities and then P2 places $k$ facilities. A point on the graph (point on an edge or a vertex) is assigned to its nearest (with respect to weighted shortest path distance) facility. In case of tie the point is assigned to the facility of P2. Each facility controls a portion of the graph which is called its service zone. Service zone of a player is the collection of service zones corresponding to its facilities. Payoff of a player is the weight of its service zone (sum of the weights of the vertices, edges, and portion of edges contained in it). The player with the larger payoff wins the game or in case where both players have same payoff the game ends in a tie.

Considering the above mentioned model we define the following problem which we call the {\it Maximum Payoff Problem}.\\\\
{\em Maximum Payoff Problem}: Given a weighted graph $G$=$(V,E)$ and a placement of $m$ facilities of P1, find a set of $k$ points $S$ on $G$ that maximizes the payoff of P2.\\\\
Throughout the paper we mainly focus on this problem. We design a polynomial time algorithm to solve the Maximum Payoff Problem on trees. Thus the main result of this paper is the following theorem.

\begin{theorem}\label{thm:opt_tree}
The Maximum Payoff Problem on trees can be solved in polynomial time.
\end{theorem}

At a high level the idea is to characterize a candidate set of polynomial size which contains a solution of the {\em Maximum Payoff Problem}. Then we design an algorithm to find $k$ points from this set which maximizes the payoff of P2. This algorithm is based on dynamic programming and runs in polynomial time.

On the other hand, we prove that the decision version of {\em Maximum Payoff Problem} is strongly $\mathcal{NP}$-complete by reducing it from the {\it Dominating Set Problem}. This implies that finding a winning strategy of P2 is $\mathcal{NP}$-complete. Consequently, we design an $1-\frac{1}{e}$ factor approximation algorithm for this problem, where $e \approx 2.718$. Lastly, we consider a different problem of finding the maximum payoff of P1 for placing $m$ facilities given that later P2 will place $k$ facilities. As a side effect of our results we obtain a lower bound on the maximum payoff of P1 for trees which is tight indeed for a special class of trees. 

The rest of the article is organized as follows. In Section \ref{sec:prob_def} we formally define the framework. In Section \ref{sec:optongraph} we charaterize the optimal solution of {\it Maximum Payoff Problem}. Then in section \ref{sec:optontree} we prove Theorem \ref{thm:opt_tree}. Section \ref{sec:com_comx} deals with the $\mathcal{NP}$-completeness proof followed by the approximation algorithm in Section \ref{sec:approximation}. We conclude our discussion with the lower bound on the maximum payoff of P1.  



\section{Problem Definitions}\label{sec:prob_def}
Let $G$=$(V,E)$ be a weighted graph embedded in $\mathbb{R}^2$. \textit{weight} of a vertex or an edge is defined by a real function $w: E\cup V \rightarrow \mathbb{R^+} \cup \lbrace 0\rbrace$. Consider an edge $e$=$(u,v)$ as a segment $\wideparen{uv}$ of length $w(e)$. Also consider a point $p$ on $e$. $p$ can be considered as a vertex of weight 0 which forms two new edges $(u,p)$ and $(p,v)$ from $e$. The weight of $(u,p)$ and $(p,v)$ are equal to the length of the segments $\wideparen{up}$ and $\wideparen{pv}$ respectively. To distinguish these new edges from the edges in $E$ we refer to them as {\it arcs}. A path between two points $p_1$ and $p_2$ is defined as a sequence of points starting with $p_1$ and ending with $p_2$ such that any two consecutive points share an edge or arc. The weight of a path is the sum of the lengths of the edges and arcs on the path. Define the distance $d(p,q)$ between two points $p$ and $q$ on $G$ (vertices or points on edges) as the length of any weighted shortest path between them. 

We consider a version of Voronoi game on $G$ played between two players P1 and P2. At first P1 chooses a set $F$ of $m$ points on $G$ to place its facilities. Thereafter P2 chooses a set $S$ of $k$ points to place its facilities, where $F \cap S$ is empty. Define the service zone of a facility $f \in F\cup S$ as,$$S(f,F\cup S)=\lbrace p: d(p,f) < d(p,f'), f' \in  (F\cup S)\setminus f\rbrace$$A point equidistant from the facilities of only one player is arbitrarily added to the service zone of one of those facilities. However, a point equidistant from the facilities of both of the players is added to the service zone of one of those facilities of P2. We note that the collection of points in the service zone of a facility can be visualized as a connected graph embedded in $\mathbb{R}^2$ which contain a subset of vertices, edges and arcs of $G$. All the points on those edges or arcs must belong to the service zone of that facility. Henceforth we consider the service zone of a facilty as a subgraph of $G$. In Figure \ref{fig:srve_zone} three facilities $f_1,f_2$ and $f_3$ are placed on a tree with unit vertex and edge weights. The service zone of $f_2$ contains the vertices $v_4,v_6$, the edge $(v_4,v_6)$, and the arcs $(p_1,v_4)$, $(v_4,p_2)$.

\begin{figure}[ht]
  \centering
    \includegraphics[height=30mm]{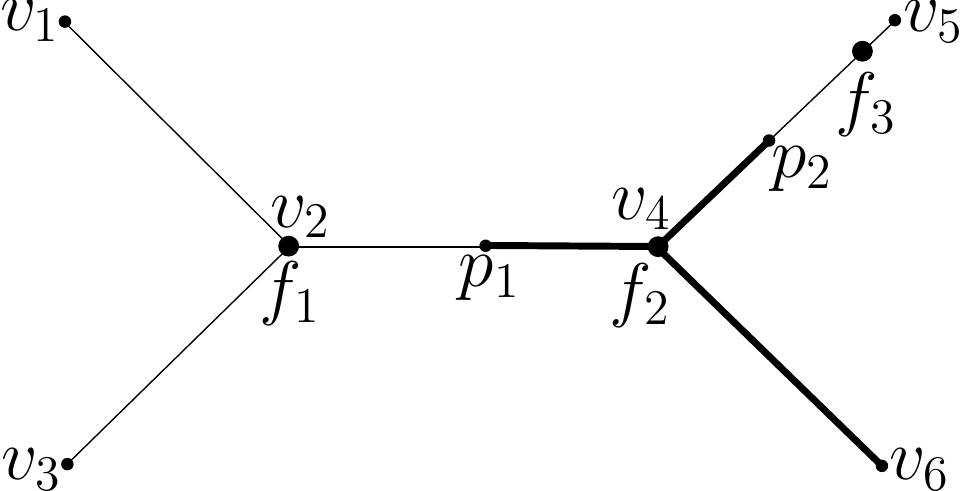}
  \caption{Service zone of $f_2$ (shown in bold)}
  \label{fig:srve_zone}
\end{figure}

For a facilty $f$, let $\mathcal{Z}(f,F\cup S)$ denote the sum of the weights of the vertices, edges and arcs in $S(f,F\cup S)$. Given two sets of facilities $F$ and $S$ placed by P1 and P2 respectively, we define the payoff of P1 as,$$\mathcal{Q}_1(F,S)=\sum_{f\in F }\mathcal Z(f,F\cup S)$$ 
Thus the payoff of P2, $\mathcal{Q}_2(F,S)$=$\mathcal{W}-\mathcal{Q}_1(F,S)$, where $\mathcal{W}$ is the sum of the weights of the vertices and edges of $G$. Given a set of facilities $F$ placed by P1 we define the maximum payoff of P2 as $\eta(F)$=$\max_{S} \mathcal{Q}_2(F,S)$, where maximum is taken over all possible $k$ facility locations of P2. Now we formally define the game framework which we call {\it One-Round $(m,k)$ Voronoi Game on Graphs}.

\begin{description}
  \item[{\it One-Round $(m,k)$ Voronoi Game on Graphs:}] Given a graph $G=(V,E)$, a weight function $w$ 
  and two players P1 and P2 interested in placing $m$ and $k$ facilities respectively, P1 chooses a set $F^*$ of $m$ 
  facility locations on $G$ following which P2 chooses a set $S^*$ of $k$ facility locations on $G$ disjoint from 
  $F^*$ such that:
    \begin{enumerate}[(i)]
      \item $\min_{F}$ $\eta(F)$ is attained at $F$=$F^*$, where the minimum is taken over all possible 
      set of $m$ facility locations $F$ of P1. 
      \item $\max_{S} \mathcal{Q}_2(F^*,S)$ is attained at $S$=$S^*$, where maximum is taken over all 
      possible set of $k$ facility locations $S$ of P2.
    \end{enumerate}
\end{description} 
Throughout the paper we consider the framework {\it One-Round $(m,k)$ Voronoi Game on Graphs}. Specifically we are interested in the optimal facility location problem for P2 on this framework which we have defined before as the {\it Maximum Payoff Problem}. With respect to this framework, given any set $F$ of $m$ facilities of P1 we are interested in finding a set of $k$ points $S$ on $G$ that maximizes $\mathcal{Q}_2(F,S)$.



\section{Characterization of Optimal Facility Locations of P2}\label{sec:optongraph}

In this section we characterize the optimal solution of the {\em Maximum Payoff Problem}, 
which is going to be used extensively in the following sections. To be precise we characterize a finite set
of points which contains an optimal solution. Note that the number of optimal solutions may be infinite. 
Figure \ref{fig:cor1} shows an example tree with unit vertex and edge weights, where One-Round $(4,1)$ Voronoi Game is played. Here optimal 
placement by P2 can be at any point on the edge $(v_1,v_2)$ making the search space for possible optimal 
locations infinite. Note that here the maximum payoff of P2 is 5 (2 from the vertices and 3 from the edges).

\begin{figure}[ht]
\centering
\includegraphics[height=30mm]{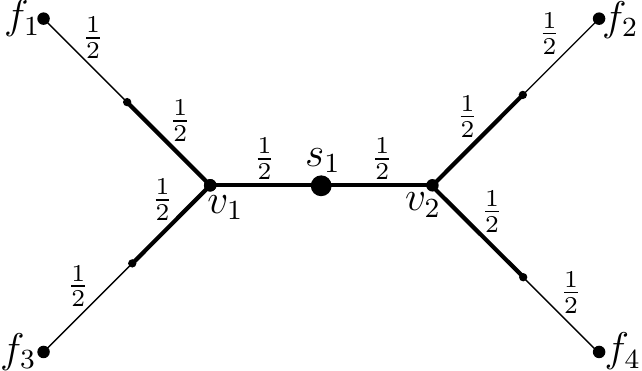}
\caption{Example of infinite optimal solutions: a possible location is $s_1$ and the corresponding service zone is shown in bold}
\label{fig:cor1}
\vspace{-0.1in}
\end{figure}

Consider any facility $f$ placed by P1 or P2. A point $b$ on $G$ is called a {\em bisector} corresponding to $f$ if there exists another facility $f'$ such that the path between $f$ and $b$ is served by $f$, the path between $b$ and $f'$ is served by $f'$, and the lengths of these paths are equal. Note that the bisectors corresponding to a facility demarcate its service zone from other facilities. As the number of paths between $f$ and any other facility is finite and each such path can contain at most one bisector, the number of such bisectors is also finite. Consider the One-Round $(4,2)$ Voronoi Game played on the graph shown in Figure \ref{fig:cor21}. $\{f_1,f_2,f_3\}$ and $\{s_1,s_2\}$ are the sets of facilities placed by P1 and P2 respectively. The service zone of $s_1$ is denoted by the bold arcs and demarcated by the bisectors $a$, $b$, $c$, $d$ and $e$. The distance of $a$ from $s_1$ is equal to the distance of $a$ from $f_3$ along the path $(s_1,v_y,a,v_2,f_3)$. The path $(a,v_2,f_3)$ is served by P1 and the path $(a,v_y,s_1)$ is served by P2.

\begin{figure}[ht]
\centering
\includegraphics[height=40mm]{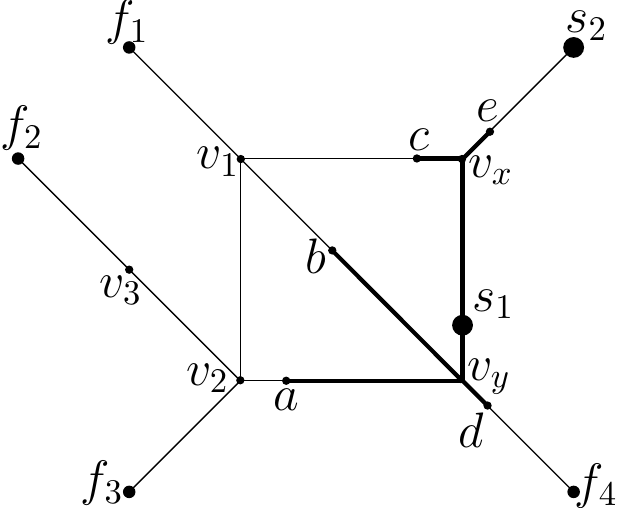}
\caption{Service zone of $s_1$ (shown in bold)}
\label{fig:cor21}
\vspace{-0.1in}
\end{figure}

Consider any placement $F$ and $S$ by P1 and P2 respectively. Observe that the bisectors corresponding to a facility $s_i$ of P2 can be of two types. The first type of bisectors split a path between $s_i$ and $f_j$, where $f_j \in F$. The second type of bisectors split a path between $s_i$ and $s_j$, where $i \neq j$ and $s_j \in S$. In Figure \ref{fig:cor21}, $a,b,c,d$ is of the first type and $e$ is of the second type of bisectors. Let us investigate that how the service zone of P2 changes if one of its facilities $s_i$ is moved along the edge on which it is lying.
Again consider the example graph in Figure \ref{fig:cor21}. Suppose the facility at $s_1$ 
is shifted by a very small distance towards $v_x$ and its new location is say $s_1'$. Subsequently bisector $a$ on the
edge $(v_2,v_y)$ shifts towards $v_y$. On the other hand, bisector $c$ on edge $(v_1,v_x)$ shifts towards $v_1$.
The bisector $e$ also shifts towards $s_2$ along the path $(s_1,v_x,s_2)$, but that does not
change the overall payoff of P2 as the path $(s_1,v_x,s_2)$ is always in the service zone of P2. Henceforth by bisector we refer to the bisectors of first type, as the other type does not contribute in computation of change in payoff.

Consider a facility $s$ $\in$ $(v_x,v_y)$ placed by P2. Also consider the bisectors corresponding to $s$ and their respective paths from $s$ to the facilities of P1. Such a path is called a $v_x$ {\it path} if it contains the vertex $v_x$ and $v_y$ does not appear in between $s$ and $v_x$ in it. A path correspond to one of those bisectors which is not a $v_x$ path is called a $v_y$ {\it path}. Now consider the two sets of bisectors $B_x$ and $B_y$ corresponding to the $v_x$ and $v_y$ paths. Note that these two sets of bisectors are not necessarily disjoint (see Figure \ref{fig:fig13}). Take two paths $\pi_1$ and $\pi_2$ corresponding to a bisector $b$ common to both of these sets. Without loss of generality say $\pi_1$ is a $v_x$ path corresponding to $s$. Then $\pi_2$ must be a $v_y$ path corresponding to $s$. Also let $f_j$ be the facility corresponding to $\pi_1$ and $\pi_2$. Note that if the facility at $s$ is shifted along the edge $(v_x,v_y)$, then the distance between $s$ and $b$ decreases. Hence after movement of $s$, P2 occupies more users from $\pi_1$ and $\pi_2$ and the payoff of P2 increases along these paths. Thus we have the following observation.

\begin{figure}[ht]
\centering
\includegraphics[height=25mm]{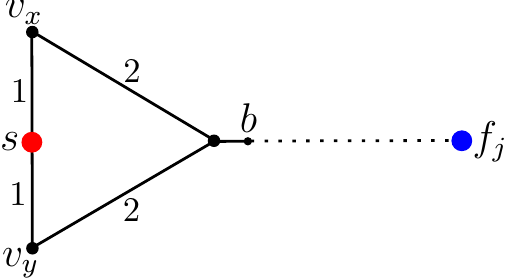}
\caption{Demonstration of common bisector $b$}
\label{fig:fig13}
\vspace{-0.1in}
\end{figure}
\begin{obs}\label{obs:common}
Consider a facility $s$ of P2 placed on $(v_x,v_y)$. If two $v_x,v_y$ paths share a common bisector, then the movement of $s$ towards $v_x$ or $v_y$ increases the payoff of P2 along these paths.  
\end{obs}

For the time being assume that $B_x \cap B_y$ is empty. Suppose a facility $s$ of P2 is shifted till the moment when one of its bisectors reaches to a vertex, say $v_l$, for the first time. If $s$ is moved further in the same direction, the number of bisectors of $s$ could be changed. Say $\epsilon$ be a distance such that if $s$ is shifted by $\epsilon$ unit from its original position, the bisectors on any path do not cross any vertex. In other words the bisectors does not change the edges on which they were lying initially. We call such an $\epsilon$ a \textit{safe distance}. Consider a path between $s$ and $f_j\in F$. If $s$ is shifted by $\epsilon$ unit along this path, then the current path between $s$ and $f_j$ shrinks by $\epsilon$ unit. Thus now the bisector on this path is shifted by $\frac{\epsilon}{2}$ unit. Similarly when $s$ is shifted away from this path, the current path between $s$ and $f_j$ expands by $\epsilon$ unit and the corresponding bisector shifts by $\frac{\epsilon}{2}$ unit. Thus we have the following observation.

\begin{obs}\label{obs:shift}
Consider a facility $s$ of P2 placed on $(v_x,v_y)$ and let $\epsilon$ be a safe distance. Suppose $s$ is shifted by $\epsilon$ unit, then each of those bisectors shifts by $\frac{\epsilon}{2}$ unit.
\end{obs}

Consider a $v_x$ path between $s$ and a facility $f_j \in F$. If $s$ is moved towards $v_x$ by a safe distance, say $\epsilon$, the path between $s$ and $f_j$ shrinks. But, considering the old path between $s$ and $f_j$ the payoff of P2 increases by $\frac{\epsilon}{2}$ unit on this path by Observation \ref{obs:shift}. As $B_x \cap B_y$ is empty the bisector corresponding to this path does not appear in $B_y$. Consider any $v_y$ path between $s$ and a facility $f_i \in F$. Then due to the movement of $s$ towards $v_x$ the path between $s$ and $f_i$ expands. Thus P2 misses a payoff of $\frac{\epsilon}{2}$ unit from the old path between $s$ and $f_i$. Thus if there are $k_1$ $v_x$ paths and $k_2$ $v_y$ paths payoff of P2 increases in $k_1$ paths and decreases in $k_2$ paths by $\frac{\epsilon}{2}$ unit. Similarly, if $s$ is shifted towards $s$ by a safe distance $\epsilon$, payoff of P2 decreases in $k_1$ paths and increases in $k_2$ paths by $\frac{\epsilon}{2}$ unit. Hence we have the following observation.

\begin{obs}\label{obs:payoff_change}
Suppose $s$ is a facility of P2 placed on $(v_x,v_y)$. Say $B_x \cap B_y$ is empty and $|B_x|=k_1$, $|B_y|=k_2$. Then if $s$ is shifted towards $v_x$ (resp. $v_y$) by a safe distance, say $\epsilon$, the payoff of P2 increases (decreases) in $k_1$ paths and decreases (increases) in $k_2$ paths by $\frac{\epsilon}{2}$ unit.
\end{obs}

Now while moving the facility $s$ suppose a bisector touches a vertex, say $v_l$, for the first time. Thus it is also the first time when $v_l$ comes in the service zone of P2 from the service zone of P1. So the payoff of P2 is increased by at least the weight of $v_l$ at that moment. This corresponding to a situation when the distance of $s$ and $v_l$ is same as the distance between $v_l$ and 
$f_l$, where $f_l$ is one of the facilities of P1 closest to $v_l$. This current location of $s$ is a transition point, when
$v_l$ moves from service zone of P1 to service zone of P2. To capture these transition points we define the following set. 
For any vertex $v_i\in V$, denote one of the facilities of P1 closest to $v_i$ by $f(v_i)$ and the 
distance between $v_i$ and $f(v_i)$ by $d_i$. Let $\Gamma(v_i)$ be the set of points in $G$ excluding $f(v_i)$ 
which are at a distance $d_i$ from $v_i$. Define $\Gamma$=$\cup_{1\leq i\leq n} \Gamma(v_i)$. It is easy to verify 
that any edge can contain at most two points of $\Gamma(v_i)$. Thus $\Gamma(v_i)$ contains $O(|E|)$ points
and consequently $\Gamma$ contains $O(|V||E|)$ points.

Let $f_t$ be any facility of P1 on any edge $(v_i,v_j)$. Then we assume that there is a point $p \in (f_t,v_j)$ very close to $f_t$ such that the distance between $p$ and $f_t$ is small enough to be considered as zero. For any such $f_t$ and $(v_i,v_j)$ that point is included into $\Gamma$ and we have the following observation.

\begin{obs}
The number of points in $\Gamma$ is $O(|V||E|+m)$.\label{obs:gamma}
\end{obs}

Consider a facility $s$ of P2 placed at a point not in $\Gamma \cup V$. Suppose $s$ is shifted along the edge in both directions until it touches a point of $\Gamma \cup V$. We show that in at least one direction the payoff of P2 increases and thus it is always beneficial to place a facility of P2 at a point of $\Gamma \cup V$. The following theorem proves this formally. 

\begin{theorem} \label{thm:optimality}
There exists a size $k$ subset of $\Gamma \cup V$ which is an optimal placement for P2.
\end{theorem}

\begin{proof}
Let $OPT_S$ be an optimal $k$ placement by P2. We construct a set $A \subseteq \Gamma \cup V$ from $OPT_S$ such that $\mathcal{Q}_2(F,OPT_S) \leq \mathcal{Q}_2(F, A)$. Suppose there is a facility $s$ at $s_t \in OPT_S$ such that $s_t \notin \Gamma \cup V$. Also let $s_t$ belongs to the edge $(v_x,v_y)$. Let 
$p_l \in$  $(v_x,s_t)$ be the point closest to $s_t$ such that $p_l \in \Gamma \cup V$. Similarly let 
$p_r \in$  $(s_t,v_y)$ be the point closest to $s_t$ such that $p_r \in \Gamma \cup V$ (see Figure $\ref{fig:fig5}$).
We show that either $\mathcal{Q}_2(F, OPT_S) \leq \mathcal{Q}_2(F, (OPT_S \setminus \{s_t\}) \cup \{p_l\})$ 
or $\mathcal{Q}_2(F,OPT_S) \leq \mathcal{Q}_2(F, (OPT_S \setminus \{s_t\}) \cup \{p_r\})$. 

\begin{figure}[ht]
\centering
\includegraphics[height=27.5mm]{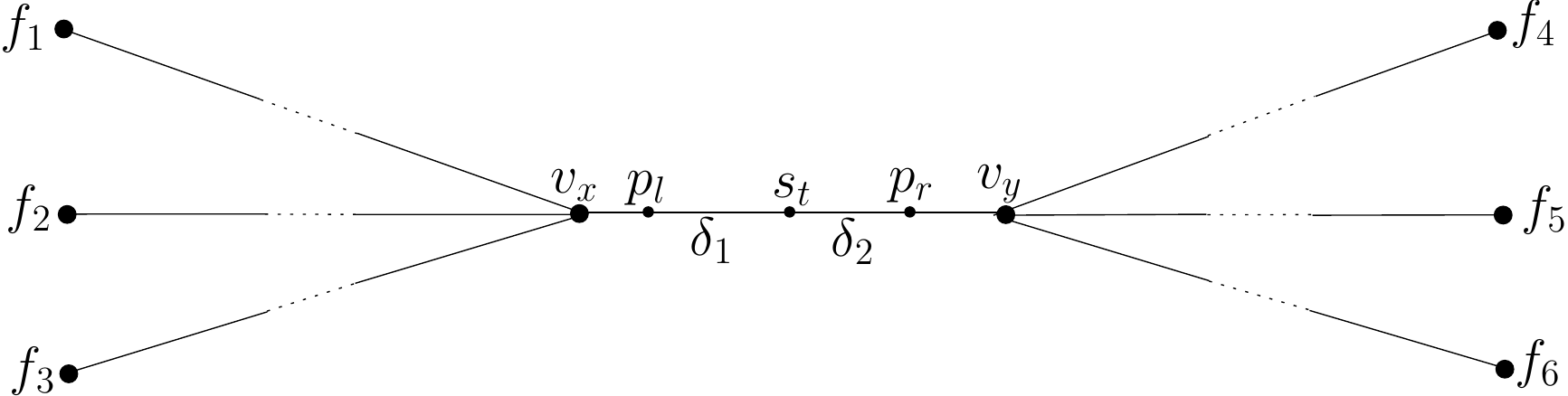}
\caption{Positions of $s_t$, $p_l$ and $p_r$} 
\label{fig:fig5}
\end{figure}

For the sake of contradiction suppose $\mathcal{Q}_2(F,OPT_S) > \mathcal{Q}_2(F, (OPT_S \setminus \{s_t\}) \cup \{p_l\})$ and
$\mathcal{Q}_2(F,$ $OPT_S) > \mathcal{Q}_2(F,$ $ (OPT_S \setminus \{s_t\}) \cup \{p_r\})$. Let the length of $(p_l,s_t)$ and 
$(s_t,p_r)$ be $\delta_1$ and $\delta_2$ respectively. Consider the $v_x$ paths and the $v_y$ paths corresponding to $s_t$ as defined before. Note that the service zone of P2 changes in only these paths when $s$ is shifted within $(p_l,p_r)$. Also consider the two sets of bisectors $B_x$ and $B_y$ corresponding to the $v_x$ and $v_y$ paths. 
Moreover, as $p_l$ and $p_r$ are the closest points of $\Gamma \cup V$ to $s_t$ the number of bisectors of $s$ remains same when $s$ is shifted till $p_l$ or $p_r$. At first consider the bisectors which are present in both $B_x$ and $B_y$. From Observation \ref{obs:common} it follows that when $s$ is shifted till $p_l$ or $p_r$ the payoff of P2 increases along the paths corresponding to these bisectors. Now consider the bisectors which are not shared by the two sets $B_x$ and $B_y$. Let $B_x$ and $B_y$ contains $k_1$ and $k_2$ such bisectors respectively. Then by Observation \ref{obs:payoff_change}, 
\begin{equation} \label{eq:solve1}
\mathcal{Q}_2(F,(OPT_S \setminus \{s_t\}) \cup \{p_l\}) \geq \mathcal{Q}_2(F,OPT_S) + (k_1-k_2)\frac{\delta_1}{2}
\end{equation}
and 
\begin{equation} \label{eq:solve2}
\mathcal{Q}_2(F,(OPT_S \setminus \{s_t\}) \cup \{p_r\}) \geq \mathcal{Q}_2(F,OPT_S) + (k_2-k_1)\frac{\delta_2}{2}
\end{equation}

Now as, $\mathcal{Q}_2(F,OPT_S) > \mathcal{Q}_2(F, (OPT_S \setminus \{s_t\}) \cup \{p_l\})$ and 
$\mathcal{Q}_2(F,OPT_S) > \mathcal{Q}_2(F, (OPT_S \setminus \{s_t\}) \cup \{p_r\})$, from Equation 
$(\ref{eq:solve1})$ and $(\ref{eq:solve2})$ we get, $(k_1-k_2)\frac{\delta_1}{2} < 0$ and $(k_2-k_1)\frac{\delta_2}{2} < 0$. 
As $\delta_1,\delta_2 > 0$, we get $(k_1-k_2) < 0$ and $(k_2-k_1) < 0$. This is a contradiction as both of $k_1$
and $k_2$ are non-negative integers. Hence the claim follows. 

We add the point $p_l$ in $A$ if $\mathcal{Q}_2(F, OPT_S) \leq \mathcal{Q}_2(F, (OPT_S \setminus \{s_t\}) \cup \{p_l\})$. Otherwise we add $p_r$ in $A$. We repeat this process to replace all such $s_t \in OPT_S$ with $s_t'$ such that $s_t' \in \Gamma \cup V$. Thus we get a set $A \subseteq \Gamma \cup V$ such that $\mathcal{Q}_2(F,OPT_S) \leq \mathcal{Q}_2(F, A)$ which completes the proof of the theorem.
\end{proof}

Note that it is sufficient to search $\Gamma \cup V$ exhaustively to get an optimal solution. But the searching time is 
still exponential in $k$.



\section{Maximum Payoff Problem on Trees: Proof of Theorem \ref{thm:opt_tree}}\label{sec:optontree}

Given a weighted tree $T$=$(V,E)$ and a set of facilities $F$=$\lbrace f_1,\ldots,f_m\rbrace$ placed by P1 on $T$ we are interested in finding a set of $k$ optimal facility locations of P2 on $T$. We'll design a polynomial time algorithm for this problem. 

Let $P$=$\{p_1,p_2,\ldots,p_\tau\}$ be any finite set of points on any tree $T'=(V',E')$. We add the points of $P\setminus V'$ into $V'$. Note that now $P$ can be regarded as a set of cut vertices, as removal of these vertices generates a finite number of subtrees (a point of $P$ can appear as a leaf in one or more subtrees). Define a \textit{partition} $T'(P)$ with respect to a finite set of points $P$ as the collection of subtrees of a tree $T'$ generated by removal of the points of $P$. 

We consider the partition $T(F)$ of $T$ (see Figure \ref{fig:part}). Let $|V|=n$. As $|F|=m$ and each point in $F$ can generate a number of subtrees equal to its degree $|T(F)|=O(m+n)$. 
We note that a facility placed by P2 in a subtree can not serve a point of another subtree, as each subtree is separated from others by facilities of P1. Thus the computation of the maximum payoff of P2 in these subtrees can be done independent of each other. Suppose the problem of placing $k'\leq k$ facilities in any such subtree is solved. Now we show how to merge those independent solutions to get a global solution for $T$. 

The problem of merging the solutions of individual subtree is similar to the \textit{Optimum Resource Allocation} problem (\cite{ora1},\cite{ora2},\cite{karush},\cite{karush53}). We have a set of $p$ resources and a set of $l$ processors. Corresponding to each processor $i$ there is an efficiency function $g_i$. $g_i(p_i)$ denotes the efficiency of $i^{th}$ processor when $p_i$ resources are allocated to it. Moreover, all the values of $g_i(p_i)$ are known for $0\leq p_i\leq p$ and $1\leq i\leq l$. The \textit{Optimum Resource Allocation} problem is to find an allocation of $p$ resources to $l$ processors so that $\sum_{i=1}^l g_i(p_i)$ is maximized, where $p_i$ resources are allocated to $i^{th}$ processor and $\sum_{i=1}^l p_i = p$. 

The following theorem implies from \cite{karush} by Karush.

\begin{figure*} 
 \centering
  \begin{minipage}[c]{0.5\textwidth}
  \centering
  \includegraphics[width=50mm] 
    {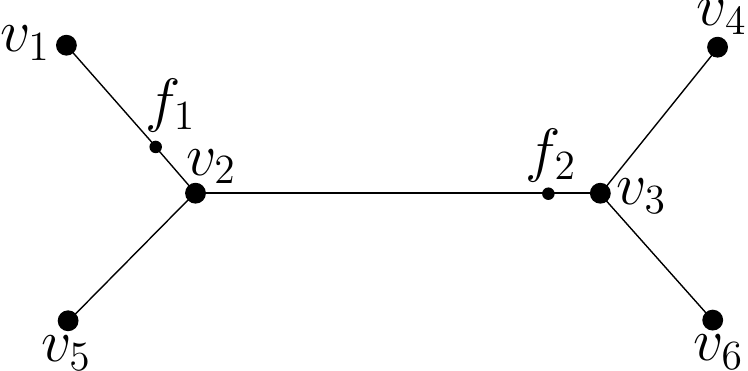}\\
    {\small (a)}\\
    \end{minipage}%
  \begin{minipage}[c]{0.5\textwidth}
  \centering
  \includegraphics[width=53mm]
    {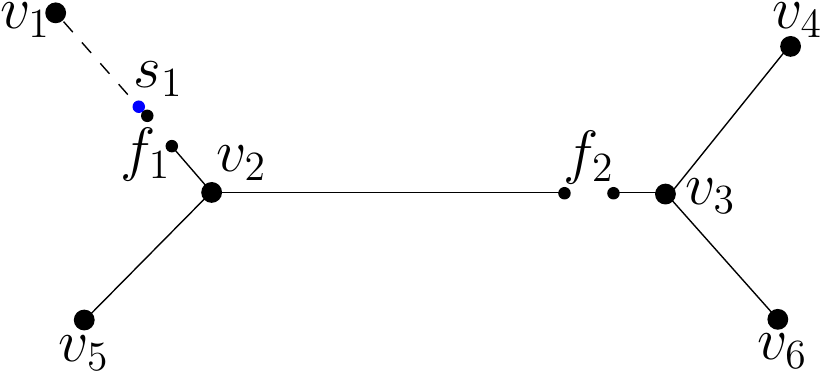}\\
    {\small (b)}\\
    \end{minipage}%
  \caption{An example tree and its partition with respect to $\lbrace f_1,f_2\rbrace$}
  \label{fig:part}
\end{figure*}

\begin{theorem}\label{th:alloc}
There is a routine ALLOC($g_1$, $\ldots$,$g_l$;$p$) which solves the Optimum Resource Allocation problem in $O(lp^2)$ time.
\end{theorem}

In this context it is worth it to mention that Hakimi {\em et al.}~\cite{Hakimi_et_al} also have used a similar routine to solve The Maximum Coverage Location Problem. Now we show how to solve our problem on $T$ by using Theorem \ref{th:alloc}. Consider the subtrees as the processors and the facilities as the resources. Denote the maximum payoff of P2 from $i^{th}$ subtree for placing $p_i$ facilities by $\mu_i(p_i)$. We set $g_i(p_i)=\mu_i(p_i), l=|T(F)|$, and $p=k$. As the payoff of P2 from $T$ is the sum of the payoffs from individual subtrees our problem is reduced to the \textit{Optimum Resource Allocation} problem. 
Thus assuming all the values of $\mu_i(p_i)$ are known, by Theorem \ref{th:alloc} it follows that the \textit{Maximum Payoff Problem} on $T$ can be solved in $O((m+n)k^2)$ time. 

%

Now we consider the problem on individual subtrees. A subtree which contains exactly one facility of P1, can be served by P2 totally by placing just one facility (see Figure $\ref{fig:part}(b)$). Now consider a subtree $T_i$ which contain at least $2$ facilities of P1. Let $\pi$ be the union of the paths of $T_i$ between the facilities of P1. Observe that $T_i \setminus \pi$ is a forest. Each tree $\lambda_j \in T_i \setminus \pi$ shares exactly one vertex with $\pi$, say $\alpha_j$. For example in Figure $\ref{fig:part}(b)$ the edge $(v_2,v_5)$ itself is such a tree and $v_2$ is the shared vertex. Note that as $\lambda_j$ does not contain any facility of P1 only one facility of P2 is sufficient to serve it totally. To be precise it is always advantageous for P2 to place a facility at $\alpha_j$ instead of placing it at any other points in $\lambda_j$. 
%
Thus for any such $\lambda_j$ we add its weight to the weight of $\alpha_j$ and remove $\lambda_j$ from $T_i$. Note that now all the leaves of $T_i$ contain facilities of P1. We refer to this kind of subtree as \textit{bounded subtree}. Hence it is sufficient to solve our problem on \textit{bounded subtrees}.

\subsection{Maximum Payoff Problem on a Bounded Subtree}
In this subsection we consider a more general problem. To avoid intricate notations we reuse some notations from before. Let $T$=$(V,E)$ be any tree where all the leaves of $T$ are occupied by facilities of P1. Each vertex has a non-negative weight. With each edge $(v_i,v_j)$ of $T$ two non-negative real values $l_{ij}$ and $w_{ij}$ are associated, where $l_{ij}$ denotes the length of the edge and $w_{ij}$ denotes the weight of that edge. Note that if an edge $e$ is within service
zone of a player, then its payoff from $e$ is equal to the weight of $e$. On the other hand, in computation of distance between two points the lengths of edges are used and weights do not play any role in this context. Service zone and payoff are defined in the same manner like before. P2 is interested in placement of $k$ facilities on the points of $\Gamma \cup V$ such that its payoff is maximized, where $\Gamma$ is the set of points on $T$ as defined in Section \ref{sec:optongraph}. 

By Theorem $\ref{thm:optimality}$ it is sufficient to consider only points of $\Gamma \cup V$ to find an optimal placement for P2 on a bounded subtree. Then the only difference between \textit{Maximum Payoff Problem} and this general problem is that in \textit{Maximum Payoff Problem} the weight and length of any edge are considered to be same, but not in the general problem. Thus if we set the same value to $w_{ij}$ and $l_{ij}$ for any $(v_i,v_j)$, then solving the general problem would solve the \textit{Maximum Payoff Problem} on any bounded subtree. Henceforth we consider the general problem. 

We propose a polynomial time algorithm for choosing $k$ optimal points from $\Gamma \cup V$. Now for each $v_i$ a point in $\Gamma(v_i)$ must lie on a path between the facility of P1 nearest to $v_i$ and another facility of P1. Thus $|\Gamma(v_i)|=O(m)$ and $
|\Gamma|=O(m|V|)$. For each point $p$ in $\Gamma \cup V$ we compute the bisectors with respect to the facilities of P1 assuming a facility of P2 is placed at $p$. Let $B$ be the set of all those bisectors. We consider the points of $\Gamma$ and $B$ also as vertices and the edges are added accordingly. As the bisectors are now vertices the service zone of any facility of P2 placed at a point of $\Gamma \cup V$ does not contain any edge partially. Let $V'=V \cup \Gamma \cup B$. Thus $V'$ is our new set of vertices. Now for each point $p \in \Gamma \cup V$ a bisector must lie on a path between $p$ and a facility of P1. Thus $O(m)$ such bisectors are possible. Hence $|B|=O(m^2|V|)$ and $|V'|=O(m^2|V|)$. 
We choose an arbitrary vertex $r\in V$ to make it the root of $T$. 

We design a routine OPT to compute the maximum payoff of P2 from $T$ for placing $k$ facilities. OPT selects $k$ vertices of $\Gamma \cup V$ in non-decreasing order of their distances from $r$ recursively. 
Suppose $v_j$ be the first vertex chosen by this routine. Let $\Gamma_{j}=T \setminus \Upsilon_{j}$, where $\Upsilon_{j}$ is the path between $r$ and $v_{j}$. As the further vertices are chosen in non-decreasing order no facilities could be placed on $\Upsilon_{j}$. Observe that $\Gamma_{j}$ is a forest. 
We need to search the subtrees in $\Gamma_{j}$ to place the remaining $k-$1 facilities. Note that these subtrees are maximal in the sense that all of their leaves contain facilities of P1. At this stage we need a routine which can optimally distribute those $k-$1 facilities to these subtrees. To resolve this issue we use the ALLOC function. We can set $g_i$ to be the maximum payoff of P2 from the $i^{th}$ subtree like before. But observe that some vertices and edges of these subtrees might already be served by the facility at $v_j$. Thus we modify the subtree by changing the weights of those vertices and edges to zero. It is to be noted that though the weights of these edges are changed to zero, their lengths remain same. Also to ensure that facilities of P2 are placed in non-decreasing order of their distances to $r$ no facility could be placed at a vertex if its distance to $r$ is less than the distance between $r$ and $v_j$. Let $V^f$ be the forbidden set of vertices of the $i^{th}$ subtree where facilities can not be placed. Then we set $g_i(p_i)$ to be the maximum payoff of P2 from the modified $i^{th}$ subtree for placing $p_i$ facilities such that no facility is placed in the vertices of $V^f$. But note that ALLOC needs the values of $g_i(p_i)$ beforehand. Thus instead of calling the routine recursively on the subtrees we ensure that the payoff values of P2 from all of those maximal subtrees are already computed. Moreover, we need a storage space where we can store all those values for future usage.  

Any maximal subtree on which the routine is executed is uniquely identified by three parameters (i) its root (ii) a subset of its vertices and edges currently served by the existing facilities of P2 and (iii) a set of forbidden vertices. We refer to these maximal subtrees as \textit{auxilliary} subtrees. OPT takes an auxilliary subtree and an integer $p$ and returns the maximum payoff of P2 from that subtree for placing $p$ facilities such that no facilities are placed at the forbidden vertices. We maintain a table $M$ to store the values returned by OPT. Each row of $M$ corresponding to an auxilliary subtree. 
$M$ contains $k$ columns marked by 1 to $k$. The entry $M[T_i,p]$ stores the maximum payoff of P2 from the auxilliary subtree $T_i$ for placing $p$ facilities avoiding the forbidden vertices. Now we define the OPT routine for any auxilliary tree $T=(V',E')$ and an integer $p$. Let $V\subseteq V'$ be the set of vertices excluding any bisector which was originally not a vertex. 

\begin{description}
\item[$OPT(T,p)$:] Say $r$ be the root of $T$. Let $V^z$ and $E^z$ be the sets of vertices and edges currently served by the existing facilities of P2. Also let $V^f$ be the set of forbidden vertices of $T$ and $\{v_{1},v_{2},\ldots, v_{t}\}$ be the vertices in $V\setminus V^f$, sorted in non-decreasing order of their distances from the root of $T$. 

\qquad If $t \leq c$ for some constant $c$, return the maximum payoff of P2 by checking all possible valid $p$ placements by P2.

\qquad For each $1\leq j\leq t$ let $E_{j}$ and $V_j$ be the respective sets of edges and vertices of $T$ served by the facility of P2 placed at $v_{j}$. Let $\Gamma_{j}=T \setminus \Upsilon_{j}$, where $\Upsilon_{j}$ is the path between $r$ and $v_{j}$. Also let $|\Gamma_{j}|=l_j$. Say $r_{ij}$, $V_{ij}$ and $E_{ij}$ be the respective root, set of vertices and set of edges of the $i^{th}$ subtree of $\Gamma_{j}$ for $1\leq i\leq l_j$. Let $V_{ij}^f$ be the set of vertices in $i^{th}$ subtree of $\Gamma_{j}$ at a distance from $r$ which is lesser than the distance between $v_j$ and $r$. Suppose $T_{ij}$ be the auxilliary tree identified by (i) the root $r_{ij}$ (ii) the respective sets of vertices and edges $(V^z \cup V_j) \cap V_{ij}$ and $(E^z \cup E_j) \cap E_{ij}$ currently served by existing facilities of P2 (iii) the set of forbidden vertices $(V^f\cap V_{ij}) \cup V_{ij}^f$. Define $g_{ij}(p_{ij})$=$M[T_{ij}, p_{ij}]$, where $\sum_{i=1}^{l_j} p_{ij}$=$p-$1.

\qquad Let $Q_{j}$=ALLOC $(g_{1j}(p_{1j}),g_{2j}(p_{2j}),\ldots,g_{l_jj}(p_{l_jj});p-1)$+$W(E_{j}+V_j)$, where $W(E_{j}+V_j)$ is the sum of the weights of the edges and vertices in $E_{j}$ and $V_j$. Lastly, set $M[T,p] = \max_{1\leq j\leq t} Q_{j}$.
\end{description}

Now we show that OPT($T,p$) indeed compute the maximum payoff of P2. If $T$ contains at most constant number of vertices at which facilities could be placed OPT($T,p$) returns the maximum value by checking all possible combinations. Otherwise, for each $v_j$ a facility is placed at $v_j$ and the remaining $p-1$ facilities are placed in the auxilliary subtrees contained in $\Gamma_{j}$. These $p-1$ facilities are allocated to $l_j$ subtrees using a call to ALLOC. Now to argue that these facilities are placed in an optimal manner we need to consider two issues (i) a facility in $T_{ij}$ does not serve any point of $T_{i'j}$ for any $i\neq i'$ and (ii) all the values $M[T_{ij}, p_{ij}]$ are available beforehand. The following observation resolves the first issue.

\begin{figure*}
  \centering
  \includegraphics[width=42mm]
    {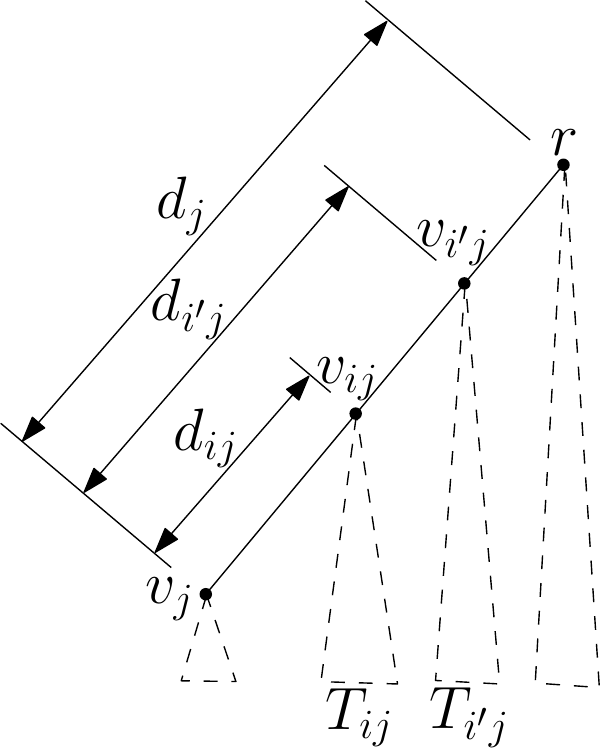}\\
  \caption{Independence of $T_{ij}$ and $T_{i'j}$}
  \label{fig:fig9}
\end{figure*}

\begin{obs}\label{obs:ind}
For any $i \neq i'$, placement of facilities of P2 in $T_{ij}$ and $T_{i'j}$ are independent of each other.
\end{obs}

\begin{proof}
Let $v_{tj}$ be the root of $T_{tj}$, where $1\leq t\leq l_j$. Consider any subtree $T_{ij}$ such that $v_{ij}$=$v_{j}$, then the service zone of any facility of P2 in $T_{ij}$ except the one at $v_j$ is limited within $T_{ij}$. Moreover, as $T_{ij}$ is connected to other subtrees through $v_{ij}$ facilities of P2 in other subtrees do not get any payoff from $T_{ij}$. Now consider two subtrees $T_{ij}$ and $T_{i'j}$ such that $v_{ij}\neq v_{j}$ and $v_{i'j}\neq v_{j}$. Let $d_{tj}$ be the distance between $v_{j}$ and $v_{tj}$ for all $t$. Also let $d_{j}$ be the distance between $r$ and $v_{j}$. Without loss of generality we assume $d_{ij} < d_{i'j}$. As all the root of the subtrees in $\Gamma_j$ are lying on the $rv_j$ path $d_{j} \geq d_{i'j}$ (see Figure $\ref{fig:fig9}$). Now the distance between any facility of P2 in $T_{i'j}$ and $v_{ij}$ is at least $d_{i'j} + (d_{i'j} - d_{ij})$ as no facility can be placed in $T_{i'j}$ within a distance $d_{i'j}$ from $v_{i'j}$. Now $d_{i'j} + (d_{i'j} - d_{ij})>d_{ij}$. Hence the facility at $v_{j}$ is closest to $v_{ij}$ than any other facilities in $T_{i'j}$. Hence any facility placed at $T_{i'j}$ does not get any payoff from $T_{ij}$. Similarly, any facility placed at $T_{ij}$ does not get any payoff from $T_{i'j}$ which completes the proof of this observation.
\end{proof}

Considering the second issue we enumerate the auxilliary subtrees in a way such that all the entries of $M$ needed by OPT($T,p$) are computed beforehand. We note that if the entries corresponding to a subtree $T'$ is needed while running OPT($T,p$), $T'$ must be a proper subtree of $T$, as $T'$ is a tree in $\Gamma_{j}$ which is obtained by deleting at least one edge of $T$. Thus it is sufficient to enumerate the auxilliary subtrees based on subtree containment relationship. In this ordering if $T'$ is contained in $T$, then $T'$ appears before $T$. We order the rows of $M$ in this manner. $M$ is filled up from top row to bottom row and in a fixed row from left to right. Hence we have the following lemma.

\begin{lemma}
OPT($T,p$) computes the maximum payoff of P2 from auxilliary tree $T$ for placing $p$ facilities such that no facilities are placed at the forbidden vertices.
\end{lemma}

To compute the maximum payoff of P2 for the original tree $T$ corresponding to the general problem we make a call to $OPT(T,k)$ where there is no existing facility of P2 and the set of forbidden vertices is empty. Then the last entry of the last row of $M$ gives the desired value.

Now we consider the time complexity of the algorithm which is precisely the product of the number of entries of $M$ and the time complexity of computing each entry. The time complexity of computation of an entry is dominated by the complexity of $t=O(|\Gamma \cup V|)=O(m|V|)$ calls to ALLOC. 
%
%
By Theorem \ref{th:alloc} each call to ALLOC needs $O(l_jp^2)=O(|V'|k^2)=O(m^2k^2|V|)$ time. Thus the total time needed is $O(m^3k^2{|V|}^2)$ to compute each entry. Now the number of entries in $M$ is the product of number of distinct auxilliary trees and size of each row ($k$). Recall that an auxilliary tree is uniquely identified by its root $r'$, a subset $U^z$ of its vertices and edges currently served by the existing facilities of P2 and a set $V^f$ of forbidden vertices. The number of distinct $r'$ is $O(|V'|)$. The way the set $U^z$ is constructed it depends on the distance of $v_j$ and $r$. Thus for a subtree with fixed root the number of distinct $U^z$ is bounded by the number of distinct distances. Now $v_j$ always belong to $\Gamma \cup V$. Thus the number of such distinct distances is $O(|\Gamma \cup V||V'|)=O(m^3{|V|}^2)$. As the set $V^f$ is also constructed based on distance the number of such $V^f$ is also $O(m^3{|V|}^2)$. Hence the number of distinct auxilliary trees is bounded by $O(m^2|V|)*O(m^3{|V|}^2)*O(m^3{|V|}^2)=O(m^8{|V|}^5)$. Thus our algorithm runs in $O(m^{11}{|V|}^7k^2)$ time and we have the following lemma.

\begin{lemma}\label{lem:bounded}
The \textit{Maximum Payoff Problem} on a bounded subtree can be solved in $O(m^{11}{|V|}^7k^2)$ time. 
\end{lemma}

Using Lemma \ref{lem:bounded} the total time needed to compute the table $M$ for all bounded trees is $O((m+n)m^{11}{n}^7k^2)$. Thus the maximum payoff of P2 from $T$ can be computed in $O((m+n)m^{11}{n}^7k^2)$ time which completes the proof of Theorem \ref{thm:opt_tree}.




\section{Computational Complexity of the Maximum Payoff Problem}\label{sec:com_comx}

This section is devoted to address the computational complexity of the {\it Maximum Payoff Problem} on graphs. To be precise we show that existence of a polynomial time algorithm for this problem is unlikely unless $\mathcal{P}\neq \mathcal{NP}$. To set up the stage, first we define the {\em decision} version of the {\it Maximum Payoff problem}.\\\\
{\it INSTANCE:} Graph $G$=$(V,E)$, a set of $m$ points $F$ on $G$, a positive real number $\delta$ and a positive integer $k$.\\
{\it QUESTION:} Does there exist a set of $k$ points $S$ on the graph $G$ such that $\mathcal{Q}_2(F,S) \geq \delta$?\\ 

Note that given a certificate for this problem consists of $G$, $F$, $S$ and $\delta$ we can verify whether the payoff of P2 is at least $\delta$ or not in polynomial time. So the problem is in $\mathcal{NP}$.  
In the remaining part of this section we show a reduction from {\it Dominating Set Problem} to this problem. As {\it Dominating Set Problem} is known to be $\mathcal{NP}$-hard \cite{gareynjohnson}, this implies that the decision version of the {\em Maximum Payoff Problem} is also $\mathcal{NP}$-hard. Let us begin our discussion by defining {\it Dominating Set} of a graph.\\
{ DOMINATING SET:} Given a graph $G$=$(V,E)$ a \textit{dominating set} is a set of vertices $S \subseteq V$ 
such that each vertex in $G$ is either in $S$ or is a neighbor of at least one vertex in $S$.

The {\it Dominating Set Problem} is as follows.\\\\
{DOMINATING SET PROBLEM}\\
{\it INSTANCE:} Graph $G$=$(V,E)$, positive integer $k \leq |V|$.\\
{\it QUESTION:} Is there a dominating set of size $k$ in $G$?\\\\
The following Theorem proves the $\mathcal{NP}$-completeness of the decision version of the {\it Maximum Payoff Problem}.

\begin{theorem}
The decision version of the Maximum Payoff Problem is $\mathcal{NP}$-complete.\label{th:nphard}
\end{theorem}
\vspace{-0.2in}
\begin{figure*}[h]
 \centering
  \begin{minipage}[c]{0.5\textwidth}
  \centering
  \includegraphics[width=50mm]
    {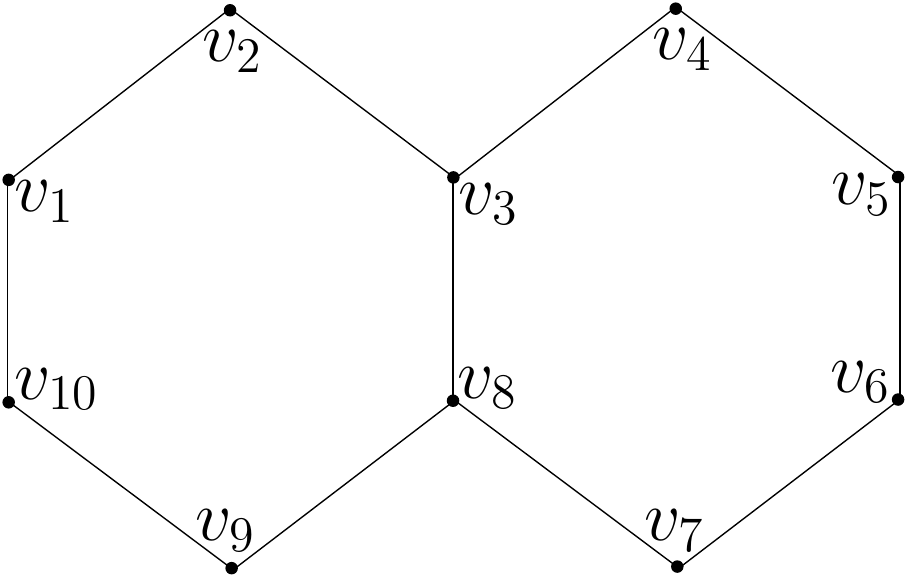}\\
    {$G$}\\
    \end{minipage}%
    \begin{minipage}[c]{0.5\textwidth}
    \centering
  \includegraphics[width=55mm]
    {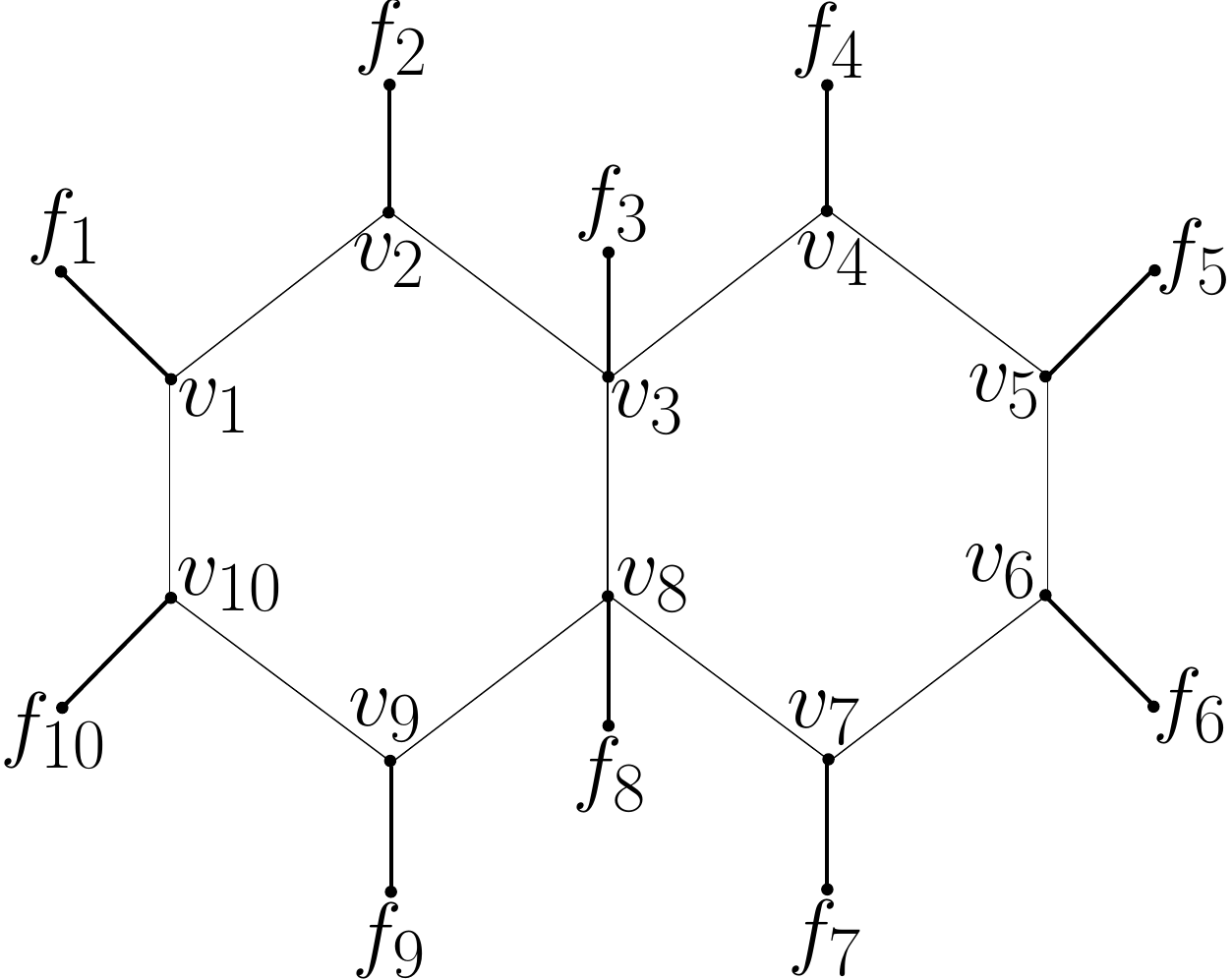}\\
    {$G'$}\\
  \end{minipage}
  \caption{Construction of $G'$ from an example graph $G$}
  \label{fig:fig7}
  \vspace{-0.2in}
\end{figure*}
\begin{proof}
It is already shown that the decision version of the {\it Maximum Payoff Problem} is in $\mathcal{NP}$. Now we show a reduction from {\it Dominating Set Problem} to this problem. Let $\mathcal{I}$=$(G,k)$ be any valid instance of {\it Dominating Set Problem}, where $G$=$(V,E)$ is an unweighted graph and $k$ is an integer. We construct a new weighted graph $G'$=$(V',E')$ from $G$ by adding a pendant vertex to each of the vertices. The construction for an example graph is shown in Figure \ref{fig:fig7}. Let $\tilde{F}$ be the set of $|V|$ new vertices. Define $V'$ = $V \cup \tilde{F}$ and $E'$ = $E \cup \{(v_i,f_i)$ $| \forall v_i \in V\}$. Assign
weight $w_e < \frac{1}{|V|+|E|+k}$ to each edge $e \in E'$ and weight $w_v=1$ to each vertex $v\in V'$. Now consider the decision version of the {\it Maximum Payoff Problem} on $G'$, where each of the points in $\tilde{F}$ contains a facility of P1. We claim that there exists a dominating set of size $k$ in $G$ if and only if there exists a set $S$ of $k$ points in $G'$ such that $\mathcal{Q}_2(\tilde{F},S)\geq |V|$. 

Consider the forward direction at first. Suppose $G$ has a dominating set $D$ of size $k$. In graph $G'$, $D$ 
can be used for placement by $P2$. Note that every vertex in $V$ is adjacent to one of the vertices of $D$. So the payoff of $P2$ is at least $|V|$.

Now consider the reverse direction. Suppose $S$ be a set of $k$ points in $G'$ such that $\mathcal{Q}_2(\tilde{F},S)\geq|V|$. Now using a construction similar in the proof of Theorem \ref{thm:optimality} we can construct a placement $S_1$ such that $S_1 \subseteq \Gamma\cup V$. Thus without loss of generality we assume $S\subseteq \Gamma\cup V$. Recall that for each edge $(f_i,v_i)$ there exists a point $p_i$ very close to $f_i$ such that distance between $p_i$ and $f_i$ is small enough to be considered as zero. Denote the set of all such points as $P$. Now observe that as weight of each edge is same $\Gamma \subseteq P\cup V$. Hence $S\subseteq P\cup V$. Now we construct a new set of placements $S'$ from $S$ in the following way. For all points $s_i\in S$, such 
that $s_i\in V$, add $s_i$ to $S'$. For all points $s_i\in S$ such that $s_i\in P$, let $s_i \in (f_i,v_j)$. Add $v_j$ to $S'$ if $v_j\notin S$, else add any vertex $v\in V$ to $S'$ such that $v\notin S$ (see Figure \ref{fig:shiftvertex}). Observe $S'\subset V$ and $\mathcal{Q}_2(\tilde{F},S')>\mathcal{Q}_2(\tilde{F},S)-kw_e$. We show that $S'$ is a dominating set. Note that the payoff $\mathcal{Q}_2(\tilde{F},S')$ can be written as $\mathcal{Q}_{E'}+\mathcal{Q}_V$, where $\mathcal{Q}_{E'}$ and $\mathcal{Q}_V$ are the sum of the weights of the respective edges and vertices in service zone of P2 corresponding to the placement $S$. Observe that $\mathcal{Q}_{E'}\leq (|V|+|E|)w_e$. Hence $\mathcal{Q}_V > \mathcal{Q}_2(\tilde{F},S)-kw_e-(|V|+|E|)w_e > \mathcal{Q}_2(\tilde{F},S)-(|V|+|E|+k)w_e$. But recall that $w_e < \frac{1}{|V|+|E|+k}$, thus $\mathcal{Q}_V > \mathcal{Q}_2(\tilde{F},S)-$1. Now the assumption was that $\mathcal{Q}_2(\tilde{F},S)\geq|V|$, which implies $\mathcal{Q}_V > |V|-$1. As $\mathcal{Q}_
V$ 
is always an integer $\mathcal{Q}_V\geq |V|$. Thus P2 serves all the vertices of $V$. Now any vertex $v_i\in V$ will be served by a facility $s_j\in S'$ if and only if $s_j$ is neighbor of $v_i$. Hence $S'$ is a dominating set of $G$ of size $k$, which completes the proof of this theorem.
\end{proof}
\begin{figure}[ht]
\centering
\includegraphics[height=10mm]{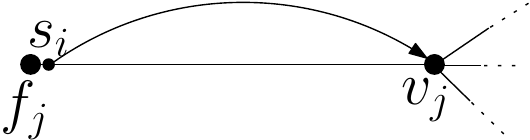}
\caption{Formation of $S'$ from $S$ in proof of Theorem \ref{th:nphard}} 
\label{fig:shiftvertex}
\end{figure}

Note that scaling of the weights of the edges and the vertices by same amount does not change the relative payoffs of P1 and P2. Thus the proof of $\mathcal{NP}$-completeness still holds if we scale up the weights of edges and vertices of the graph used in our construction by a factor of $|V|+|E|+k$. Now the weight of any edge is at most $1$ and the weight of any vertex is $|V|+|E|+k$. Thus the problem remains $\mathcal{NP}$-complete even if the weights of the graph is bounded by a polynomial in the length of the input. Hence the decision version of the Maximum Payoff Problem is strongly $\mathcal{NP}$-complete indeed.



\section{Approximation Bound for the Maximum Payoff Problem on Graphs}\label{sec:approximation}\vspace{0.1in}
In this section we discuss an $1-\frac{1}{e}$ factor approximation algorithm for the {\it Maximum Payoff Problem}. We show
a construction for generating an instance of the {\em Weighted Maximum Coverage Problem} from an instance of the {\em Maximum Payoff Problem} in polynomial time
and use the existing approximation algorithm for the {\it Weighted Maximum Coverage Problem} to get an approximation algorithm for our problem. But before that let us define the {\it Weighted Maximum Coverage Problem}.\\\\
{\it Weighted Maximum Coverage Problem} (WMCP): Given an universe $X=\{x_1,x_2,\ldots,x_n\}$, a family $\mathcal{S}$ 
of subsets of $X$, an integer $\tau$ and weight $w_i$ associated with each $x_i\in X$, find $\tau$ subsets from $\mathcal{S}$ such that total weight of the covered elements is maximized.\\

WMCP is known to be $\mathcal{NP}$-hard and there is an $1-\frac{1}{e}$ factor greedy approximation algorithm for it, where $e\approx 2.718$ \cite{Hochbaum}. In each iteration this algorithm chooses a subset, which contains the maximum weighted uncovered elements. Thus we have the following theorem.

\begin{theorem}
\cite{Hochbaum} The greedy algorithm for WMCP achieves an approximation ratio of $1-\frac{1}{e}$.\label{th:hochbaum}
\end{theorem}

\begin{figure}[ht]
\centering
\includegraphics[height=40mm]{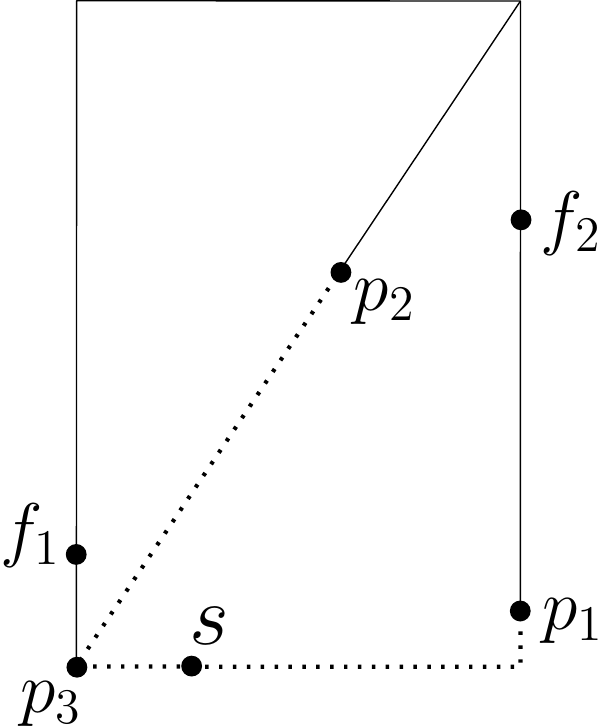}
\caption{Service zone of $s$}
\label{fig:app}
\vspace{-0.1in}
\end{figure}

Let $G$=$(V,E)$ be any graph and $F$ be any set of facilities placed by P1 in $G$. P2 wants to place $k$ new facilities. 
Now from Theorem \ref{thm:optimality} it follows that there exists an optimal placement by \textit{P2} which is a subset of $\Gamma\cup V$. Now consider any placement of facility at $s\in \Gamma\cup V$ by P2. Let $B_s$ be the set of bisectors corresponding to $s$. For example in Figure \ref{fig:app}, P1 has placed two facilities at $f_1$ and $f_2$ and P2 has placed a facility at $s$. The service zone of P2 is shown with dotted lines. Here the set $B_s$ will be equal to $\{p_1,p_2,p_3\}$. 
Define $$B= \{\cup_{s\in \Gamma\cup V} B_s\}\cup \Gamma$$ 
It is easy to see that each edge of $G$ can contain at most a constant number of bisectors corresponding to each point of $\Gamma \cup  V$. Thus from Observation $\ref{obs:gamma}$ it follows that the cardinality of $B$ is bounded by $O({(\Gamma\cup V)}E)$= $O((V+E)^2)$. Now we construct a new graph $G'$=$(V',E')$ from $G$, where $V'$=$V\cup B\cup F$. For any edge $e_{ij}\in E$ with end vertices $v_i$ and $v_j$, which does not contain any point of $B$, include $e_{ij}$ in $E'$. Any edge $e_{ij}$ which contains one or more points of $B$, say $\{b_1,b_2,\ldots ,b_l\}$, sorted along $v_i$ to $v_j$, add the edges $(v_i,b_1),(b_1,b_2),\ldots, (b_{l-1},b_l)$ to $E'$. The weight of each such edge is equal to its length. Now observe that service zone of the facility of P2 placed at a point of $\Gamma\cup V$ is a subgraph whose edges are in $E'$ and vertices are in $V\cup B$.

Now consider the set system where $X$ is equal to $E' \cup V$ and for each point $p_i\in \Gamma\cup V$ define the set 
$S_i \subseteq E' \cup V$ such that $S_i$ is the set of edges and non-bisector vertices which are in service zone of the facility of P2 at $p_i$. Now run the greedy algorithm for the {\it Weighted Maximum Coverage Problem} on this set system with $\tau$=$k$. The weight returned by this algorithm is the payoff of P2. Thus we have the following lemma.

\begin{lemma}
Any $\alpha$ factor approximation algorithm for WMCP produces an $\alpha$ factor approximation for the {\it Maximum Payoff Problem}.\label{lemma:alpha}
\end{lemma}

Thus by combining Theorem \ref{th:hochbaum} and Lemma \ref{lemma:alpha} we have the following theorem.

\begin{theorem}\label{thm:approx}
There exists an $1-\frac{1}{e}$ factor approximation algorithm for the Maximum Payoff Problem.
\end{theorem}



\section{Bound on Maximum Payoff of P1 on Trees}\label{sec:tree}
We are given a tree $T$=$(V,E)$. We show a lower bound on maximum payoff of P1 from $T$. Denote the total weight of $T$ by $\mathcal{W}$. Recall the definition of partition of a tree from Section \ref{sec:optontree}. We show that there is a collection of points $P$ in $T$ such that the weight of each subtree in the partition $T(P)$ is at most $\frac{\mathcal{W}}{|P|+1}$. Here we assume that if a subtree in a partition contains a vertex of $P\cap V$, then its weight is changed to zero. 

\begin{lemma}
For any tree $T$ and a positive integer $\tau$ there is a set of points $P=\{p_1,p_2,\ldots,p_\tau\}$ which partitions $T$ into at 
least $\tau+1$ subtrees such that weight of each subtree in $T(P)$ is at most $\frac{\mathcal{W}}{\tau+1}$.\label{lem:tau}
\end{lemma}

\begin{proof}
Observe that it is enough to show that for any weighted tree $T$=$(V,E)$ with weight function $w$ and a positive integer $\tau$ there is a point $\mathring{p}$ which partitions the tree into two or more parts so that the weight of one part is at most $\frac{\tau \mathcal{W}}{\tau+1}$ and the weight of every other part is at most $\frac{\mathcal{W}}{\tau+1}$. 

Choose an arbitrary vertex of tree as the root of $T$. Define an extended weight function $w_T: V \rightarrow \mathbb{R}^+ \cup \{0\}$ such that for a vertex $v_i$,\\  
\[ w_T(v_i) = \left\{ 
  \begin{array}{l l}
    w(v_i) &  \textrm{ if $v_i$ is a leaf}\\
    w(v_i) + \sum_{j:\text {$v_j$ is a child of $v_i$}}(w_T(v_j)+w(v_i,v_j)) & \text{ otherwise}
  \end{array} \right.\]


Now observe that there will always be a vertex with extended weight greater than or equal to $\frac{\mathcal{W}}{\tau+1}$ 
and all of its children are having extended weight less than $\frac{\mathcal{W}}{\tau+1}$. Denote that vertex by 
$\breve{v}$. Let the children of $\breve{v}$ be $\{v_1,v_2,\ldots,v_l\}$. Now if for all $1\leq i\leq l$, 
$w_T(v_i)+w(\breve{v},v_i)$ is less than $\frac{\mathcal{W}}{\tau+1}$, then $\mathring{p}=\breve{v}$. 
Otherwise there exists a child $v_j$ of $\breve{v}$ such that $w_T(v_j)+w(\breve{v},v_j)>\frac{\mathcal{W}}{\tau+1}$, 
and $w_T(v_j)<\frac{\mathcal{W}}{\tau+1}$. However, in that case there exists a point $p$ 
on the edge $(\breve{v},v_j)$, which partitions the tree into two parts, one having weight $\frac{\mathcal{W}}{\tau+1}$ 
and the other having weight $\frac{\tau\mathcal{W}}{\tau+1}$. Thus $\mathring{p}=p$ and the result follows.

\end{proof}
The next corollary follows from Lemma \ref{lem:tau}.

\begin{cor}
There exists a placement strategy of P1 such that it always achieves at least $\frac{m-k+1}{m+1}\mathcal{W}$ 
as its payoff for {\it One-Round (m,k) Voronoi Game on} $T$.\label{cor:boundonpayoff}
\end{cor}

\begin{proof}
We prove this corollary by proposing a placement strategy of P1. By Lemma $\ref{lem:tau}$ we know that there
exists a set $F'$ which partition the tree $T$ in a manner such that each of the partition is having
weight at most $\frac{\mathcal{W}}{m+1}$, where $|F'|$=$m$. Suppose P1 places its facilities on the points
of $F'$. By placing $k$ facilities P2 can occupy at most $k$ partitions. Payoff of P2 in that case would be at most $\frac{\mathcal{W}}{m+1}k$ . Hence the payoff of P1 is at least $\frac{m-k+1}{m+1}\mathcal{W}$, which completes
the proof of this corollary.
\end{proof}

Now consider a restricted version of this game where P2 places only one facility. Also consider 
the complete bipartite graph $K_{1,m}$ with $m$ edges of equal weight. In this case, an optimal 
strategy of P1 is to place a facility at the \textit{central vertex} (i.e., at the vertex with degree $m$) and the remaining $m-1$ facilities anywhere on the graph. On the other hand, P2 chooses a point as close as possible to the \textit{central vertex} as its optimal strategy. Thus service zone of P2 is limited within an edge and payoff of P1 is $\frac{m}{m+1}\mathcal{W}$. So, the 
bound of Corollary $\ref{cor:boundonpayoff}$ is tight for $k=1$. 



\section{Conclusion}
Considering the optimal facility location problem for P1 we have shown a lower bound on the maximum payoff. But the status of this problem is still unresolved for general graphs. Also it is not known that whether this problem could be solved in polynomial time on trees. 

We have shown that the {\it Maximum Payoff Problem} is strongly $\mathcal{NP}$-complete and subsequently designed an $1-\frac{1}{e}$ factor approximation algorithm. But no tight lower bound is known on the maximum payoff of P2. We have designed a polynomial time algorithm for the {\it Maximum Payoff Problem} on tree. However, the time complexity of this algorithm is very high and thus one might be interested to reduce it. On the other hand, it would be interesting to study the nature of this problem for some special classes of trees.


\bibliographystyle{plain}
\bibliography{voronoi_game_ref}

\begin{thebibliography}{10}

\bibitem{Ahn_et_al}
Hee-Kap Ahn, Siu-Wing Cheng, Otfried Cheong, Mordecai~J. Golin, and Ren{\'e}
  van Oostrum.
\newblock Competitive facility location: the voronoi game.
\newblock {\em Theor. Comput. Sci.}, 310(1-3):457--467, 2004.

\bibitem{CFL1}
Simon~P. Anderson.
\newblock Equilibrium existence in the linear model of spatial competition.
\newblock {\em Economica}, 55(220):479--491, 1988.

\bibitem{BandyapadhyayBDS13}
Sayan Bandyapadhyay, Aritra Banik, Sandip Das, and Hirak Sarkar.
\newblock Voronoi game on graphs.
\newblock In {\em WALCOM}, pages 77--88, 2013.

\bibitem{ora1}
A.~Charnes and W.~W. Cooper.
\newblock The theory of search optimal distribution of effort.
\newblock {\em Management Science 5}, pages 44--49, 1958.

\bibitem{Cheong_One-Round_Voronoi_Game}
Otfried Cheong, Sariel Har-Peled, Nathan Linial, and Jir\'{\i} Matousek.
\newblock The one-round voronoi game.
\newblock {\em Discrete {\&} Computational Geometry}, 31(1):125--138, 2004.

\bibitem{ora2}
J.~De~Guenni.
\newblock Optimum distribution of effort: An extension of the koopman basic
  theory.
\newblock {\em J ORSA 9}, pages 1--7, 1961.

\bibitem{nash_DurrT07}
Christoph D{\"u}rr and Nguyen~Kim Thang.
\newblock Nash equilibria in voronoi games on graphs.
\newblock In {\em ESA}, pages 17--28, 2007.

\bibitem{CFL4}
H.A. Eiselt and G.~Laporte.
\newblock Competitive spatial models.
\newblock {\em European J. Oper. Res.}, 39:231--242, 1989.

\bibitem{CFL_survey2}
H.A. Eiselt, G.~Laporte, and J.F. Thisse.
\newblock Competitive location models: a framework and bibliography.
\newblock {\em Transportation Sci.}, 27:44--54, 1993.

\bibitem{Fekete_one-round_Voronoi_game}
S{\'a}ndor~P. Fekete and Henk Meijer.
\newblock The one-round voronoi game replayed.
\newblock {\em Comput. Geom.}, 30(2):81--94, 2005.

\bibitem{nash_transitive_FeldmannMM09}
Rainer Feldmann, Marios Mavronicolas, and Burkhard Monien.
\newblock Nash equilibria for voronoi games on transitive graphs.
\newblock In {\em WINE}, pages 280--291, 2009.

\bibitem{gareynjohnson}
M.~R. Garey and David~S. Johnson.
\newblock {\em Computers and Intractability: A Guide to the Theory of
  NP-Completeness}.
\newblock W. H. Freeman, 1979.

\bibitem{CFL6_Hakimi}
S.L. Hakimi.
\newblock On locating new facilities in a competitive environment.
\newblock {\em European J. Oper. Res.}, 12:29--35, 1983.

\bibitem{CFL3_Hakimi}
S.L. Hakimi.
\newblock Location with spatial interactions: competitive location and games,
  in: R.l. francis, p.b. mirchandani (eds.).
\newblock {\em Discrete Location Theory}, pages 439--478, 1990.

\bibitem{Hochbaum}
Dorit~S. Hochbaum.
\newblock {\em Approximation algorithms for $\mathcal{NP}$-Hard problems}.
\newblock PWS publishing company, 1996.

\bibitem{karush}
William Karush.
\newblock A general algorithm for the optimal distribution of effort.
\newblock {\em Management Science}, 9(1):pp. 50--72, 1962.

\bibitem{Vgameonpath}
Masashi Kiyomi, Toshiki Saitoh, and Ryuhei Uehara.
\newblock Voronoi game on a path.
\newblock {\em IEICE Transactions}, 94-D(6):1185--1189, 2011.

\bibitem{karush53}
Bernard~O. Koopman.
\newblock The optimum distribution of effort.
\newblock {\em Journal of the Operations Research Society of America}, 1(2):pp.
  52--63, 1953.

\bibitem{CFL5_Hakimi}
M.~Labb{\'e} and S.L. Hakimi.
\newblock Market and locational equilibrium for two competitors.
\newblock {\em Oper. Res.}, 39:749--756, 1991.

\bibitem{nash_cycle_MavronicolasMPS08}
Marios Mavronicolas, Burkhard Monien, Vicky~G. Papadopoulou, and Florian
  Schoppmann.
\newblock Voronoi games on cycle graphs.
\newblock In {\em MFCS}, pages 503--514, 2008.

\bibitem{Hakimi_et_al}
Nimrod Megiddo, Eitan Zemel, and S~Louis Hakimi.
\newblock The maximum coverage location problem.
\newblock {\em SIAM Journal on Algebraic Discrete Methods}, 4(2):253--261,
  1983.

\bibitem{CFL2}
A.~Okabe and M.~Aoyagy.
\newblock Existence of equilibrium configurations of competitive firms on an
  infinite two-dimensional space.
\newblock {\em J. Urban Econom.}, 29:349--370, 1991.

\bibitem{CFL_duopoly}
Daniela Saban and Nicolas Stier-Moses.
\newblock The competitive facility location problem in a duopoly: Connections
  to the 1-median problem.
\newblock In {\em Proceedings of the 8th International Conference on Internet
  and Network Economics}, WINE'12, pages 539--545, Berlin, Heidelberg, 2012.
  Springer-Verlag.

\bibitem{CFL_tree1}
Shogo Shiode and Zvi Drezner.
\newblock A competitive facility location problem on a tree network with
  stochastic weights.
\newblock {\em European Journal of Operational Research}, 149(1):47--52, 2003.

\bibitem{Voronoi_game_on_graphs}
Sachio Teramoto, Erik~D. Demaine, and Ryuhei Uehara.
\newblock The voronoi game on graphs and its complexity.
\newblock {\em J. Graph Algorithms Appl.}, 15(4):485--501, 2011.

\bibitem{CFL_survey1}
R.L. Tobin, T.L. Friesz, and T.~Miller.
\newblock Existence theory for spatially competitive network facility location
  models.
\newblock {\em Ann. Oper. Res.}, 18:267--276, 1989.

\end{thebibliography}

\end{document}